\documentclass[12pt,a4paper]{amsart}
\usepackage{amsfonts,amsmath,epsfig,amscd,amssymb,latexsym}
\usepackage{amsmath,amssymb,verbatim}
\usepackage{color}

\newtheorem{lem}{Lemma}
\newtheorem{proposition}{Proposition}

\newtheorem*{definition}{Definition}
\theoremstyle{definition}

\newtheorem*{remark}{Remark}

\newcommand{\Q}{\mathbb{Q}}
\newcommand{\C}{\mathbb{C}}

\newcommand{\R}{\mathbb{R}}

\newcommand{\ZZ}{\mathbb{Z}}
\newcommand{\OO}{\mathcal{O}}

\providecommand{\abs}[1]{\ensuremath{\left\lvert #1 \right\rvert}}

\providecommand{\norm}[1]{\ensuremath{\left\Vert #1 \right\Vert}}

\providecommand{\vv}[1]{\textquotedblleft #1\textquotedblright}

\DeclareMathOperator*{\Vol}{Vol}

\newtheorem{lemma}{Lemma}
\newtheorem{rem}{Remark}
\newtheorem{prop}{Proposition}
\newtheorem{defi}{Definition}
\newtheorem{theorem}{Theorem}
\newtheorem{cor}{Corollary}

\DeclareMathOperator*{\argmin}{argmin}

\title{Algebraic lattice codes for linear fading channels}

\author{Laura Luzzi and Roope Vehkalahti}

\begin{document}

\maketitle

\section{Introduction}
In the decades following Shannon's work, the quest to design codes for the additive white Gaussian noise (AWGN) channel led to the development of a rich theory, revealing a number of beautiful connections between information theory and geometry of numbers. 
One of the most striking examples is the connection between classical lattice sphere packing and the capacity of the AWGN channel. 
The main result  states that any family of lattice codes with linearly growing \emph{Hermite invariant} achieves a constant gap to capacity. These classical results and many more can be found in the comprehensive book by Conway and Sloane \cite{CS}. 


The early sphere packing results suggested that lattice codes could achieve the capacity of the AWGN channel and led  to a series of works trying to prove this, beginning with \cite{DeBuda} and finally completed in \cite{Erez_Zamir}. 
Thus, while there are still plenty of interesting questions to consider, the theory of lattice codes for the single user AWGN channel is now well-established.

However, although the AWGN channel is a good model for deep-space or satellite links, modern wireless communications require to consider more general channel models which
include  time or frequency varying fading and possibly multiple transmit and receive antennas.  Therefore in the last twenty years 
coding theorists have focused on the design of 
lattice codes for multiple and single antenna fading channels \cite{BERB, TSC}.

 However, the question of whether lattice codes can achieve capacity in fading channels has only been addressed recently. 
The first work that we are aware is due to 
S. Vituri 
\cite[Section 4.5]{Vituri},
and gives  a proof of existence of lattice codes achieving a constant gap to capacity for i.i.d SISO channels.  It seems 
that with minor modification this proof is enough 
to guarantee
the existence of capacity achieving lattices. In the single antenna  i.i.d fading channels this problem was considered also in \cite{HNGLOBECOM} and in our paper \cite{ISIT2015_SISO}.

In \cite{Cong} it was shown 
that polar lattices achieve capacity in single antenna  i.i.d fading channels. This is not only an existence result, but does also give an explicit code construction. In \cite{CLB} the authors prove existence of lattice codes achieving capacity for the compound MIMO channel, where  the fading is random during the first $s$ time units, but then gets repeated in blocks of length $s$. This work is most closely related to \cite{OE}, which was considering a similar question.

In   \cite{LV2014,LV2017} we proved that lattice codes derived from number fields and division algebras do achieve a constant gap to capacity over single and multiple antenna fading channels. As far as we know this was the first result achieving constant gap to MIMO capacity with lattice codes. In \cite{HN} the authors corrected and generalized \cite{HNGLOBECOM} and improved on our gap in the case of Rayleigh fading MIMO channels.

However while in our work \cite{LV2017}  the gap to capacity  is relatively large, these codes are 
\emph{almost universal} in the sense that a \emph{single} code 
achieves a constant gap to capacity for \emph{all} stationary ergodic fading channel models satisfying a certain condition for fading \eqref{WLLN}.
With some limitations this gap is also uniform 
for all such channels (see Remark \ref{uniformity}).

In this work we are revisiting  some of the results in \cite{LV2017} and presenting them from a slightly different and more general perspective.
Our approach is based on generalizing the classical sphere packing approach to fading channels.  In \cite{LV2017} we 
introduced the concept of \emph{reduced Hermite invariant}  of a lattice with respect to a linear group of block fading matrices. 
As a generalization of the classical result for AWGN channels, we  
 proved that if  a family  of lattices has linearly growing reduced Hermite invariant, 
 it achieves a constant
 gap to capacity in the 
 block fading MIMO channel. In this work 
 we 
 extend this result 
 and show that given any linearly fading channel model we can define a corresponding notion of reduced Hermite invariant.
 We also prove that in some cases the reduced Hermite invariant of a lattice is actually a \emph{homogeneous minimum} with respect to \emph{homogeneous form} (which depends on the fading model). From this perspective the classical   sphere packing result  \cite[Chapter 3]{CS} is just one  example of the  general 
 connection between linear fading channels and the homogeneous minima of the corresponding forms.

In Section \ref{general_linear_fading}, we begin by defining a general linear fading model, which captures several channels of interest for practical applications. In Section \ref{lattices}, we recall how to obtain a finite signal constellation from an infinite lattice under an average power constraint. 
In Section \ref{AWGN}, we 
review  how the classical Hermite invariant can be used as a design criterion to build capacity approaching lattice codes in AWGN channels. In Section \ref{reduced_section} we generalize the concept of Hermite invariant to linear fading channels by introducing the general reduced Hermite invariant. We will also show that  replacing the Hermite invariant with the reduced one as a code design criterion leads to an analogous capacity result in linear fading channels.

In Section \ref{diagonal} we 
focus on channels where the fading matrices are diagonal. This brings us to consider ergodic fading  single antenna channels.  Following \cite{LV2017}, we 
then show how  lattice constructions from algebraic number fields
can be used to approach capacity in such channels. We begin by considering lattices arising
from the canonical embedding of the
ring of algebraic integers, then examine the question of improving the gap to capacity using non-principal ideals of number fields \cite{ISIT2015_SISO}\footnote{More precisely, the ideal lattice construction was considered in the extended version of \cite{ISIT2015_SISO}, available at http://arxiv.org/abs/1411.4591v2.}.
In particular, we  show that our information-theoretic problem  is  actually equivalent to a certain classical problem in algebraic number theory.

Finally in Section \ref{geometry} we extend  the results in \cite{LV2017} and show that in many relevant channel models the reduced Hermite invariant of a lattice is actually a homogeneous minimum of a certain form.

\section{General linear fading  channel} \label{general_linear_fading}
In this work we consider complex vector-valued channels, where the transmitted (and received) elements are vectors in $\C^k$.
A \emph{code}  $\mathcal{C}_k$  is  a finite  set of elements in  $\C^k$. We assume that both the receiver and the transmitter know the 
code.

Given a matrix $H\in M_n(\C)$ and a vector $x\in \C^k$, in order to hold on the tradition that
a transmitted vector is a row, we 
introduce the notation
$$
H[x]=(H(x^T))^T.
$$

Let us assume we have an infinite sequence of random matrices $H_k$, $k=1,2\dots, \infty$, where for every $k$, $H_k$ is a $k\times k$  matrix. Given such sequence of matrices we can define a corresponding channel.
Given an input $x=(x_1,\dots, x_k)$, we will write the channel output as 
\begin{equation}\label{channelmodel}
y=H_k[x]+w,
\end{equation}
where $w$ is a  a length $k$ random vector, with i.i.d complex Gaussian entries
with variance 1  and zero mean, and $H_k$ is a random matrix representing fading. We assume that the receiver always knows the channel  realization of $H_k$ and is 
trying to guess which was the transmitted codeword $x$ based on $y$ and $H_k$. This set-up 
defines a
\emph{linear fading channel} (with channel state information at the receiver), where the term \vv{linear} simply refers to the fact that the fading can be represented as the action of a linear transform on
the transmitted codeword. This type of channels  (but without channel state information) have been considered 
before in 
\cite{WF}.

In the following sections we consider the problem of designing codes for this type of channels.
In the remainder of the paper we will assume the extra condition that
the determinant of the random matrices $H_k$ is non-zero with probability one.
The channel  model under consideration captures many 
communication channels of practical significance.
For example  when 
$H_k$ 
is a deterministic identity matrix,
we have the classical additive Gaussian channel. On the other hand if 
$H_k$ is a diagonal matrix with
i.i.d Gaussian random elements
with zero mean, we obtain the Rayleigh fast fading channel.
Finally, if  
$H_k$ is a block diagonal matrix we obtain a block
fading MIMO channel.


\section{Lattices and finite codes} \label{lattices}
As mentioned previously, our finite codes  $\mathcal{C}$ are simply subsets of elements  in  $\C^k$.
We consider the ambient space  $\C^k$ as a metric space with the Euclidean
norm.

\begin{defi}\label{euclidean norm}
Let $v=(v_1, ...,v_k)$ be a vector in $\C^k$.
The \emph{Euclidean norm}   of $v$ is $||v||=\sqrt{\sum_{i=1}^k |v_i|^2}$.
\end{defi}

Given a \emph{transmission power} $P$, we require that every codeword  $x\in \mathcal{C}\subset \C^k$ satisfies the average power constraint
\begin{equation} \label{power_constraint}
\frac{1}{k} \norm{x}^2\leq P.
\end{equation}
The \emph{rate} of the code is given by $R=\frac{\mathrm{log}_2\abs{\mathcal{C}}}{k}$. 
 
In this work we focus
on finite codes $\mathcal{C}_k$ that are derived from \emph{lattices}.

A  full {\em lattice} $L \subset \C^k$ has the form $L=\ZZ b_1\oplus \ZZ b_2\oplus \cdots \oplus \ZZ b_{2k}$, where the vectors $b_1,\dots, b_{2k}$ are linearly independent over $\R$, i.e., form a lattice basis.

Given an average power constraint
$P$, the following Lemma suggests that by shifting a lattice and considering its intersection with the $2k$-dimensional ball $B(\sqrt{kP})$ of radius $\sqrt{kP}$, we can have codes having roughly $\Vol( B(\sqrt{kP}))$ elements.  

\begin{lemma}[see \cite{GL}]\label{shift}
Suppose that $L$ is a full lattice in  $\C^k$  and  $S$ is a Jordan-measurable bounded subset of $\C^k$. Then there exists $x\in \C^k$ such that
$$
|(L+x)\cap S|\geq\frac{\mathrm{Vol}(S)}{\mathrm{Vol}(L)}.
$$
\end{lemma}

Let 
$\alpha$ be an 
energy normalization constant and $L$ a $2k$-dimensional lattice in $\C^k$ satisfying $\Vol(L)=1$. According to  Lemma \ref{shift}, we can choose an element $x_R\in \C^k$ such that for the code
\begin{equation} \label{spherical_shaping}
\mathcal{C}=B(\sqrt{kP} )\cap( x_R+\alpha L)
\end{equation}
we have the cardinality bound
\begin{equation}\label{guarantee}
\abs{\mathcal{C}}\geq\frac{\Vol(B(\sqrt{kP}))}{\Vol (\alpha L)}=\frac{C_k P^k}{\alpha^{2k}},
 \end{equation}
where $C_k=\frac{(\pi k)^k}{k!}$. 
We can now see that given  a lattice $L$ with $\Vol(L)=1$, the number of codewords we are guaranteed to get only depends on the size of 
$\alpha$.

From now on, given a lattice $L$ and power limit $P$, the finite codes we are considering 
will always satisfy \eqref{guarantee}.
We note that while the finite codes are not subsets of the scaled lattice $\alpha L$, 
they inherit many properties from the underlying lattice. 



\section{Hermite invariant  in the AWGN  channel} \label{AWGN}

In this section we will present the classical Hermite invariant approach 
to build
capacity approaching codes for the AWGN channel   \cite[Chapter 3]{CS}.
We remark that this   channel can  be seen as an example of our general set-up \eqref{channelmodel} by assuming that for every $k$ the random matrix $H_k$ is a $k\times k$ identity matrix with probability one.
The channel equation can now be written as 
$$
y=  x + w,
$$
where $x \in\mathcal{C}_k \subset \C^k$  is the transmitted codeword and $w$ is the Gaussian noise vector.

After the transmission, the receiver tries to guess which was the transmitted codeword $x$ by performing maximum likelihood (ML) decoding, and outputs
$$\hat{x}=\argmin_{\bar{x} \in \mathcal{C}_k} \norm{y-\bar{x}}=\argmin_{\bar{x} \in \mathcal{C}_k} \norm{x-\bar{x}+w}.$$

 This suggests a simple code design criterion to minimize the error probability. Given a power limit  $P$, the codewords of 
 $\mathcal{C}_k$  should be as far apart as possible. As the properties of the finite code $\mathcal{C}_k$ are inherited from the underlying lattice, we should give a reasonable definition of what it means that lattice points are far apart.

\begin{defi}\label{Hermite}
The \emph{Hermite invariant} of a $2k$-dimensional  lattice $L_{k}\subset \C^k$  
is defined as 
$$
\mathrm{h}(L_k )=\frac{\mathrm{inf}\{\,||x||^2 \mid  x\in L_k , x\neq 0\}}{\Vol(L_k)^{1/k}},
$$
where $\Vol(L_k)$ is the volume of the fundamental parallelotope of the lattice $L_k$.
\end{defi}

\begin{theorem}\label{Gaussian}
Let $L_k \subset \C^k$ be a family of $2k$-dimensional  lattice codes satisfying $\mathrm{h}(L_k)\geq 2kc$.
Then any  rate 
$$
R<   \log_2 P  -\log_2 \frac{2}{\pi e c }
$$
is achievable using the lattices $L_{k}$ 
with ML decoding.
\end{theorem}
 
\begin{proof}
Given a power limit $P$, 
recall that the finite codes  we are considering are of the form
$\mathcal{C}=B(\sqrt{kP} )\cap( x_R+\alpha L_k)$. Without loss of generality, we can assume that $\Vol(L_k)=1$. Here 
$\alpha$ is a power normalization constant that we will soon solve and which will define the achievable rate. The minimum distance in the received constellation is
$$d=\min_{\substack{x, \bar{x} \in \mathcal{C}\\ x \neq \bar{x}}} \norm{x- \bar{x}  }.$$
The error probability is upper bounded by
$$P_e \leq \mathbb{P}\left\{ \norm{w}^2 \geq \left(\frac{d}{2}\right)^2\right\}.$$ 
Note that we can lower bound the minimum distance as follows:
$$
d^2 \geq \alpha^2 \min_{x \in L_k \setminus \{0\}} \norm{ x}^2\geq \alpha^2 \mathrm{h}(L_k)\geq\alpha^2 2ck.
$$

Therefore we have the upper bound
\begin{equation} \label{upper}
P_e \leq \mathbb{P}\left\{ \norm{w}^2 \geq \frac{\alpha^2 ck  }{2}    \right\}.
\end{equation}
Let $\epsilon >0$. Since $2\norm{w}^2$ is a $\chi^2$ random variable with $2k$ degrees of freedom, due to the law of large numbers, 
\begin{equation} \label{w}
 \lim_{k \to \infty} \mathbb{P}\left\{\frac{\norm{w}^2}{k} \geq 1+\epsilon\right\}= \lim_{k \to \infty} \mathbb{P}\left\{\frac{2\norm{w}^2}{2k} \geq 1+\epsilon\right\} \to 0
 \end{equation}
Assuming $\alpha^2 = \frac{2(1+\epsilon)}{c}$, we then have that $P_e \to 0$ when $k \to \infty$, and the cardinality bound \eqref{guarantee} implies that
$$
|\mathcal{C}|\geq \frac{C_k P^k}{\alpha^{2k}}=\frac{C_k P^k c^k}{2^{k}(1+\epsilon)^k}.
$$
For large $k$,
$C_k  \approx \frac{(\pi e)^k}{\sqrt{2\pi k}}$ using Stirling's approximation. 

It follows that $\forall \epsilon>0$ we can achieve rate
$$
R= \log_2 P - \log_2 \frac{2(1+\epsilon)}{\pi e c}.
$$
Since $\epsilon$ is arbitrary, this concludes the proof.
\end{proof}

\section{ Hermite invariant in general linear fading model}\label{reduced_section}
In the previous section we saw how the Hermite invariant can be used as a design criterion to build capacity approaching codes in the AWGN channel. 
Let us now 
define 
a generalization of this invariant for
linear fading channels.

Suppose we have an infinite sequence of random matrices $H_k$, $k=1,2\dots, \infty$, where $H_k$ is a $k\times k$ matrix.
Given an input $x=(x_1,\dots, x_k)$, we will write the channel output as 
$$
y=H_k[x]+w,
$$
where $w$ is a length $k$ random vector, with i.i.d complex Gaussian 
entries with variance $1$ per complex dimension.
We assume that the receiver knows the realization of $H$.

Given a channel realization $H$, the receiver outputs the ML estimate 
$$\hat{x}=\argmin_{\bar{x} \in \mathcal{C}} \norm{H[x] + w -H[\bar{x}]}.$$

 From the receiver's perspective 
 this is equivalent to decoding the code 
$$
H[\mathcal{C}]=\{H[x]\mid x\in \mathcal{C}\}
$$
over an AWGN channel.

As we assumed that the finite codes are of the form (\ref{spherical_shaping}), we have
$$
H[\mathcal{C}]\subset \{H[x]\mid x\in x_{R}+\alpha L \}=\{z\mid  z\in H[x_R] +\alpha H[L]\},
$$
where 
$$
H[L]=\{H[x]\mid x\in L\}.
$$
We can now see that the properties of $H[\mathcal{C}]$ are inherited from  the set $H[L]$. 

\smallskip

If we assume that the matrix $H$  has full rank with probability $1$, then the linear mapping $x \mapsto H[x]$ is a bijection of $\C^k$ onto itself with probability $1$. 

Assuming 
that 
$L_{k}\subset\C^k$ has basis $\{b_1,\dots ,b_{2k}\}$ we 
have that
$$
H[L_{k}]=\{H[x]\mid x\in L_{k}\}=\ZZ H[b_1] \oplus \cdots \oplus \ZZ H[b_{2 k}],
$$
is a full-rank lattice with basis $\{H[b_1],\dots, H[b_{2 k}]\}$. Since it is full-rank, we know that $h(H[L_{k}])>0$, but 
is it possible to choose $L_k$ in such a way that 
$h(H[L_{k}])$ would be large irrespective of the channel realization $H$?
 Let us now try to formalize this idea.
 
We can write  the random matrix $H_k $ in the  form
$$
H_k= |\mathrm{det}(H_k)|^{1/k} H_k'
$$
where  $ |\mathrm{det}({H_k}' )|=1$. 
Clearly, if the term $|\mathrm{det}(H)|^{1/k}$
happens to be small, it will crush the Euclidean distances of points in $H[L_{k}]$.  However, 
we will show that if the random matrices $H_k$ are \vv{well behaving}, then
it is possible to design lattices that are robust against 
fading. 


\begin{defi}\label{reduced_Hermite_invariant}
Let $\mathcal{A}$ be a set of invertible matrices such $ \forall A\in \mathcal{A}$, $|\det(A)|=1$. 
The  \emph{reduced Hermite invariant} \cite{LV2017} of a $2k$-dimensional  lattice $L\subset \C^k$ with respect to 
$A$  is defined as 
$$
\mathrm{rh}_{\mathcal{A}}(L)=\underset{A\in \mathcal{A}}{\mathrm{inf}} \{\mathrm{h}(A[L])\}.
$$
\end{defi}

It is easy to see that 
\begin{equation}\label{switch}
\underset{A\in \mathcal{A}}{\mathrm{inf}} \Big\{\underset{x\in L,\; x\neq 0}{\mathrm{inf}} ||A[x]||^2 \Big\}=
\underset{x\in L,\; x\neq 0}{\mathrm{inf}} \Big\{\underset{A\in \mathcal{A}}{\mathrm{inf}}||A[x]||^2 \Big\}.
\end{equation}

This observation 
suggests the following definition.
\begin{defi}
We call 
$$
||x||_{\mathcal{A}}=\mathrm{inf}\{\,||A[x]|| \mid  A\in \mathcal{A} \},
$$
the \emph{reduced norm} of the vector $x$ with respect to  the set $A$.
\end{defi}

With this observation we realize that 
\begin{equation} \label{rh_characterization}
\mathrm{rh}_{\mathcal{A}}(L )=\frac{\mathrm{inf}\{\,||x||_{\mathcal{A}}^2 \mid  x\in L ,\; x\neq 0\}}{\Vol(L)^{1/k}}.
\end{equation}
If the set $\mathcal{A}$ includes the identity matrix, we obviously have 
$$
\mathrm{rh}_{\mathcal{A}}(L )\leq \mathrm{h}(L).
$$


Suppose that $\{H_k\}_{k\in \mathbb{N}^+}$ is a fading process such that $H_k \in M_{k\times k}(\C)$ is full-rank with probability $1$, and suppose that the weak law of large numbers holds for the random variables $\{\log \det (H_kH_k^{\dagger})\}$, i.e. $\exists \mu >0$ such that $\forall \epsilon >0$,
\begin{equation} \label{WLLN}
\lim_{k \to \infty} \mathbb{P}\left\{ \abs{\frac{1}{k}  \log \det(H_kH_k^{\dagger})-\mu}>\epsilon\right\}=0.
\end{equation}

We denote the set  of all 
invertible realizations of  $H_k$ with 
$\mathcal{A}_k^*$. 
Then define
\begin{equation}\label{normalized}
\mathcal{A}_k=\{ |\mathrm{det}(A)|^{-1/k}A\mid A\in \mathcal{A}_k^*\}.
\end{equation}



\begin{theorem}\label{main}
Let $L_k \subset \C^k$ be a family of $2k$-dimensional  lattice codes satisfying $\mathrm{rh}_{\mathcal{A}_k}(L_k)\geq 2kc$, and suppose that the channel satisfies (\ref{WLLN}).
Then any  rate 
$$
R < \log_2 P + \mu -\log_2 \frac{2}{\pi e c}
$$
is achievable using the codes $L_{k}$ with ML decoding. 
\end{theorem}
 
\begin{proof}
Given a power constraint $P$, recall that we are considering finite codes of the form (\ref{spherical_shaping}), where 
$\alpha$ is a power normalization constant that we will soon solve. 

The minimum distance in the received constellation is
$$d_{H}=\min_{\substack{x, \bar{x} \in \mathcal{C}\\ x \neq \bar{x}}} \norm{H[x- \bar{x}]  }\geq \min_{\substack{x \in L_k \\ x \neq 0}} \norm{H[\alpha x]  },$$
and by the hypothesis on the reduced Hermite invariant,  
$$
d_{H}^2 \geq \alpha^2 \min_{x \in L_k \setminus \{0\}} \norm{ H[x]}^2\geq \alpha^2 \det(HH^{\dagger})^{1/k} \mathrm{rh}_{\mathcal{A}_k}(L_k)\geq\alpha^2 \det(HH^{\dagger})^{1/k}2ck.
$$

The ML 
error probability is
bounded by
$$P_e \leq \mathbb{P}\left\{ \norm{w}^2 \geq \left(\frac{d_{H}}{2}\right)^2\right\}.$$ 

Fixing $\epsilon>0$, the law of total probability implies that 
\begin{align*}
&P_e \leq \mathbb{P}\left\{\frac{d_H^2}{4k} \geq 1+ \epsilon\right\}\mathbb{P}\left\{ \norm{w}^2 \geq \frac{d_H^2}{4}\;\Big |\; \frac{d_H^2}{4k} \geq 1+ \epsilon\right\} + \mathbb{P}\left\{ \frac{d_H^2}{4k}<1+\epsilon \right\}  \\
& \leq \mathbb{P}\left\{ \frac{\norm{w}^2}{k} \geq 1 + \epsilon\right\}+\mathbb{P}\left\{ \frac{d_H^2}{4k}<1+\epsilon \right\} \\
& \leq \mathbb{P}\left\{ \frac{\norm{w}^2}{k} \geq 1 + \epsilon\right\}+\mathbb{P}\left\{ \frac{\alpha^2 c\det(HH^{\dagger})^{1/k}}{2} < 1+\epsilon \right\}
\end{align*}
Recall that the first term tends to zero when $k \to \infty$ due to (\ref{w}). 
The second term will tend to zero as well if we choose 
$$
\log_2\left(\frac{2(1+\epsilon)}{\alpha^2 c}\right)=\mu-\delta
$$
for some $\delta>0$. 
 Equation \eqref{guarantee} gives us that

$$
R=\frac{1}{k}\log_2 \abs{\mathcal{C}} \leq \log_2 P - \log_2 \frac{\alpha^2}{C_k}
$$
For large $k$, $C_k\approx \frac{(\pi e)^k}{\sqrt{2\pi k}}$.
It follows that we can achieve rate
$$
R= \log_2 P + \mu -\delta -\log_2 \frac{2(1+\epsilon)}{\pi e c} 
$$
Since $\epsilon$ and $\delta$ are arbitrary, any rate 
$$R < \log_2 P + \mu -\log_2 \frac{2}{\pi e c}$$
is achievable.
\end{proof}

\section{Code design for diagonal fading channels}\label{diagonal}
Let us now consider a fading channel where for every $k$ we have $H_k=\mathrm{diag}[h_1,h_2,\dots,h_k]$.
Assume that each $h_i$ is non-zero with probability 1 and that $\{h_i\}$ forms an  ergodic stationary random process.
In this model, sending a single symbol $x_i$ during the ith time unit leads to the channel equation
\begin{equation}\label{single_channel}
y_i=h_i \cdot x_i + w_i,
\end{equation}
where  $w_i$ is a zero-mean Gaussian complex random variable with variance $1$.

The corresponding set of matrices  $\mathcal{A}_k$ in \eqref{normalized} is a  subset of the set of  diagonal matrices in $M_k(\C)$ having determinant with absolute value 1. 

The assumption that the process $\{h_i\}$ is ergodic and stationary implies that each of the random variables $h_i$ have equal statistics. Therefore we can simply use  $h$ to refer to the statistics of all $h_i$. Assuming now also that $\sum_{i=1}^k \frac{1}{k} \log|h_i|^2$ converges in probability to some constant,  we have  the following.

\begin{cor}\label{singlecapacity}
Suppose that we have a family of lattices $L_k\subset \C^k$, where $\mathrm{rh}_{\mathcal{A}_k}(L_k)\geq 2kc$. Then any rate 
$$
R< \mathbb{E}_h \left[ \log_2 P \abs{h}^2\right] -\log_2 \frac{2}{\pi e c}
$$
is achievable with the family $L_k$ over the fading channel \eqref{single_channel}. 
\end{cor}

\begin{proof}
This statement follows immediately from Theorem \ref{main}, where $\mu=\mathbb{E}_h\left[\log_2\abs{h}^2\right]$.
\end{proof}


Given two sets $\mathcal{A}_k' \subseteq \mathcal{A}_k$, we have for any lattice $L$ that
$$
\mathrm{rh}_{\mathcal{A}_k'}(L)\geq \mathrm{rh}_{\mathcal{A}_k}(L).
$$
From now on, we will fix $\mathcal{A}_k$ to be the set of all diagonal matrices in $M_k(\C)$ having determinant with absolute value $1$. 
Note that with this choice, if $\mathrm{rh}_{\mathcal{A}_k}(L_k)\geq 2kc$ then Corollary \ref{singlecapacity} holds for any channel of the form \eqref{single_channel}.


Let $(x_1,x_2, \dots, x_k)\in \C^k$. According to \cite[Proposition 8]{LV2017}\footnote{More precisely, this result is slightly stronger than the statement of Proposition 8, but it is clear from its proof.} we 
have 
\begin{equation}\label{reduced}
||(x_1,\dots, x_k)||_{\mathcal{A}_k}^2= k|x_1\cdots x_k|^{2/k}.
\end{equation}

We can now see that 
a lattice with large reduced Hermite invariant must have the property that
the product of  the coordinates of any non-zero element of the lattice is large.

\begin{defi}\label{tulo}
Given $x=(x_1,...,x_k) \in \C^k$, we define  its \emph{ product norm} as
$
n(x)=\prod_{i=1}^k |x_i|.
$
\end{defi}


\begin{defi}
Then the \emph{normalized product distance} of $L_k$ is
\begin{equation}\label{prodred}
\mathrm{Nd}_{\mathrm{p, min}}(L_k)=\inf_{\mathbf{x} \in L_k \setminus \{0\}} \frac{n(x)}{\Vol(L_k)^{\frac{1}{2}}}.
\end{equation}
\end{defi}

Combining \eqref{reduced}, \eqref{rh_characterization} and \eqref{prodred} we have that
\begin{equation}\label{connection}
\mathrm{rh}_{\mathcal{A}_k}(L_k)= k(\mathrm{Nd}_{\mathrm{p, min}}(L_k))^{2/k}.
\end{equation}

This result gives us a more concrete characterization
of the reduced Hermite invariant
and suggests possible candidates for good lattices.

\subsection{Codes from algebraic number fields}
The product distance criterion in the previous section had already been derived in \cite{BERB} by analyzing the pairwise error probability in the special case where the process $\{h_i\}$ is i.i.d Gaussian. 
The authors also pointed out that lattices that are derived
from number fields have large product distance.  
We will now shortly present this classical construction and then study how close to the capacity we can get using number fields. For the relevant background on algebraic number theory we refer the reader to \cite{Nark}.



Let $K/\Q$ be  a totally complex extension of degree $2k$ and $\{\sigma_1,\dots,\sigma_k\}$ be a set  of  $\Q$-embeddings, such that we have chosen one from each complex conjugate pair. Then we can define a
\emph{relative canonical embedding} of $K$ into $\C^k$ by
$$
\psi(x)=(\sigma_1(x),\dots, \sigma_k(x)).
$$

The following lemma is a basic result from algebraic number theory.
\begin{lemma}
The ring of algebraic integers $\OO_K$ has a  $\ZZ$-basis $W=\{w_1,\dots ,w_{2k}\}$ and $\{\psi(w_1),\dots ,\psi(w_{2k})\}$ is a $\ZZ$-basis for the full  lattice $\psi(\OO_K)$ in $\C^k$.
\end{lemma}

For our purposes the key property of the lattices  $\psi(\OO_K)$ is that for any non-zero  element $\psi(x)=(\sigma_1(x),\dots, \sigma_k(x))\in \psi(\OO_K)$, we have that
$$
\Big|\prod_{i=1}^k \sigma_i(x)\Big|^2=nr_{K/\Q}(x)\in \ZZ,
$$
where $nr_{K/\Q}(x)$ is the algebraic norm of the element $x$.
In particular it follows that $|\prod_{i=1}^k \sigma_i(x)|\geq 1$. 

We now know that $\psi(\OO_K)$ is a $2k$-dimensional lattice in $\C^k$ with the property that
$\mathrm{Nd}_{\mathrm{p, min}}(\psi(\OO_K))\neq 0$ and therefore $\mathrm{rh}_{\mathcal{A}_k}(\psi(\OO_K))\neq 0$. This is true for any totally complex number field. Let us now show how the value of    $\mathrm{rh}_{\mathcal{A}_k}(\psi(\OO_K))$  is related to
an algebraic invariant of the field $K$.

 We will denote the \emph{discriminant} of a number field $K$ with $d_K$. For every number field it is a non-zero integer.

The following Lemma states some
well-known results from  algebraic number theory and a translation of these results into our coding-theoretic language.
\begin{lemma}\label{complex}
Let  $K/\Q$ be a totally complex extension of degree $2k$ and  let $\psi$ be the relative canonical embedding. Then
$$
\mathrm{Vol}(\psi(\mathcal{O}_K))=2^{-k}\sqrt{|d_K|} 
$$
$$
\mathrm{Nd}_{\mathrm{p, min}}(\psi(\mathcal{O}_K))=\frac{2^{\frac{k}{2}}}{|d_K|^{\frac{1}{4}}} \,\, \mathrm{and}\,\,   \mathrm{rh}_{\mathcal{A}_k}(\psi(\mathcal{O}_K))=\frac{2k}{|d_K|^{1/2k}}.
$$
\end{lemma}
We have now translated the question of finding 
algebraic lattices with the largest reduced Hermite invariants into the task of finding the totally complex number fields with the smallest discriminant. Luckily this is 
a well-known mathematical problem with a tradition of almost a hundred years.

In \cite{Martinet}, J. Martinet proved the existence of an infinite tower of totally complex number fields $\{K_k\}$ of degree $2k$, where $2k=5\cdot2^t$, such that
\begin{equation} \label{G}
 \abs{d_{K_k}}^{\frac{1}{k}}=G^2,
\end{equation}
for $G \approx 92.368$.
For such fields $K_k$ we have that
$$
\mathrm{Nd}_{\mathrm{p, min}}(\psi(\mathcal{O}_{K_k}))=\left(\frac{2}{G}\right)^{\frac{k}{2}} \,\, \mathrm{and}\,\,\,   \mathrm{rh}_{\mathcal{A}_k}(\psi(\mathcal{O}_{K_k}))=\frac{2k}{G}.
$$

Specializing Corollary \ref{singlecapacity} to the family of lattices $L_k=\psi(\mathcal{O}_{K_k})$ 
derived from Martinet's tower, which satisfy the hypothesis with $c=1/G$, we then have the following result:
\begin{prop}\label{Martinetcapacity}
Finite codes drawn from the lattices $L_k$ achieve any rate satisfying
$$
 R < \mathbb{E}_h \left[ \log_2 P \abs{h}^2\right] -\log_2 \frac{2G}{\pi e}.
$$
\end{prop}

\begin{rem}\label{uniformity}
We note that given a stationary and ergodic fading process $\{h_i\}$ the capacity of the corresponding channel is
$$
C= \mathbb{E}_h \left[ \log_2(1+P|h|^2)\right].
$$ 
It is easy to prove that the  rate achieved in Proposition \ref{Martinetcapacity} is a constant gap from the capacity.
This gap is also universal in the following sense. Let us consider all ergodic  stationary channels  with the same first order statistics for $h$. Then  the \emph{same} sequence of finite codes achieve the same gap to capacity in all the channels simultaneously. 

\end{rem}

\begin{rem}
We note that the number field towers we used are not the best known possible.  It was shown in \cite {Hajir_Maire} that one can construct a family of totally complex fields such that $G<82.2$, but this choice would add some notational complications.
\end{rem}

\begin{rem}
The families of number fields  on which our constructions are based 
were first brought to coding theory in \cite{LT}, where the authors pointed out that the corresponding lattices  have linearly growing Hermite constant.  This directly implies that they are only a constant gap from the AWGN  capacity.
C. Xing in \cite{Xing} remarked that   these families of number fields provide the best known normalized product distance. 
Overall number field lattices in fading channels  have been well-studied in the literature. However, to the best of our knowledge we were the first to prove that they actually do achieve a constant gap to capacity over fading channels. 
\end{rem}

\subsection{Codes from ideals}

As seen in the previous section, lattice codes arising from the rings of algebraic integers of number fields with constant root discriminants will achieve a constant gap to capacity over fading channels. 
  However, known lower bounds for discriminants  \cite{Odlyzko} 
  imply that no matter which number fields we use, 
the gap cannot be reduced beyond a certain threshold (at least when using our current approach to bound the error probability). 
It is then natural to ask whether
other lattice constructions could lead us closer to capacity.
The most obvious generalization is to consider additive subgroups of $\OO_K$ and in particular  ideals of $\OO_K$, which 
will have non-zero reduced Hermite invariant.  
Most works concerning 
 lattice codes from number fields focused on either the ring $\OO_K$ or a principal ideal  $a\OO_K$;
 a more general setting was considered in \cite{FOV} and \cite{Oggier}, which addressed 
   the question of increasing the normalized product distance using non-principal ideals $I$.

The problem with this approach is that while  finding the reduced Hermite invariant of  lattices $\psi(\OO_K)$ or  $\psi(a\OO_K)$ is an easy task, the same is not true for $\psi(I)$ when $I$ is non-principal. We will now show how this problem can be reduced to another well-known problem in algebraic number theory and how it can be used to study the  performance limits of the lattices $\psi(I)$. Here we will follow the extended arXiv version of \cite{ISIT2015_SISO}.

We note that while the  concept of reduced Hermite invariant is more general and 
its information-theoretical meaning is clearer,  number theoretic proofs are easier when using the equivalent product distance notation. Therefore we will mostly 
focus on the product distance in this section. 

Let $K$ be  a totally complex field of degree $2k$. We will use the notation $\mathrm{N}(I)=[\OO_K:I]$ for the norm of an ideal $I$.
 From classical algebraic number theory we  have that $N(a\OO_K)=|nr_{K/\Q}(a)|$ and 
$N(AB)=N(A)N(B)$.

\begin{lemma}\label{idealvolume}
Suppose that $K$ is a totally complex field of degree $2k$ and that  $I$ is an integral ideal in $K$.
 Then $\psi(I)$ is
a $2k$-dimensional lattice in $\C^k$ and 
$$
\Vol{(\psi(I))}=[\OO_K:I] 2^{-k}\sqrt{|d_K|}.
$$
\end{lemma}
This well-known result 
allows us to compute the volume of an ideal, but 
computing its normalized product distance 
is a more complicated issue. In \cite[Theorem 3.1]{FOV} the authors stated the analogue of the following  result for the totally real case. It is simply a restatement of the definitions.
\begin{prop} 
Let us suppose that $K$ is a totally complex field of degree $2k$ and that $I$ is an integral ideal of $K$. We then have that
\begin{equation}
\mathrm{Nd}_{\mathrm{p, min}}(\psi(I))=\frac{2^{\frac{k}{2}}}{|d_K|^{\frac{1}{4}}}\mathrm{min}(I),
\end{equation}
where $\mathrm{min}(I):=\underset{x \in I \setminus \{0\}}{\mathrm{min}}\sqrt{\frac{|\mathrm{nr}_{K/\Q}(x)|}{\mathrm{N}(I)}}.$
\end{prop}
\begin{proof}
This  result follows from Lemma \ref{idealvolume}, the definition of the normalized product distance and from noticing that
 $\sqrt{|\mathrm{nr}_{K/\Q}(x)|}=|n(\psi(x))|.$
\end{proof}

Due to the basic ideal theory of algebraic numbers $\mathrm{min}(I)$ is always larger or equal to $1$. If $I$ is not a principal ideal then we have that
$\mathrm{min}(I)\geq \sqrt{2}$. Comparing this to Lemma \ref{complex} we find that, given a non principal ideal domain $\OO_K$, we  should use an ideal $I$, which is not principal, to maximize the product distance. Now there are two obvious questions. Given a  non principal ideal domain $\OO_K$, which ideal $I$ should we use and how much can we gain? 
Before answering these questions we need the following.

\begin{lemma}\cite{FOV}
For any non-zero element $x \in K$, 
$$
\mathrm{Nd}_{\mathrm{p, min}}(\psi(xI))=\mathrm{Nd}_{\mathrm{p, min}}(\psi(I)).
$$
\end{lemma}
This result proves that every ideal in a given ideal class has the same normalized product distance. It follows that given a ring of integers $\OO_K$, it is enough to 
consider
one ideal from every ideal class.
Given an ideal $I$ we will denote with $[I]$ the ideal class to which
$I$ belongs.

Let us denote with $N_{\text{min}}(K)$  the norm of an ideal $A$ in $K$ with the property that every ideal class of $K$ contains an integral ideal with norm $N(A)$ or smaller.  The question of finding the 
size of  
 $N_{\text{min}}(K)$ is a classical problem in algebraic number theory. We refer the reader to \cite{Zim} for further reading. The following result is from the extended arXiv version of \cite{ISIT2015_SISO}.

\begin{prop}\label{idealform}
Let us suppose that $K$ is a totally complex number field of degree $2k$  and that $I$ is an ideal that maximizes the normalized product distance over  all ideals in $K$. We then have that
$$
\mathrm{Nd}_{\mathrm{p, min}}(\psi(I))= \frac{2^{k/2}\sqrt{N_{\text{min}}(K)}}{|d_K|^{\frac{1}{4}}} \, \,\mathrm{and} \,\,\, \mathrm{rh}_{A_k}(\psi(I))=\frac{2k(N_{\text{min}}(K))^{1/k}}{|d_K|^{\frac{1}{2k}}}.
$$
\end{prop}
\begin{proof}
Let $L$ be any 
ideal in $K$, and 
 suppose that $A$ is an integral ideal in the class  $[L]^{-1}$ with the smallest norm. We then have that there exists an element $y \in \OO_K$ such that $y\OO_K=AL$. As $n(\psi(y))=\sqrt{N(L)N(A)}$ and $N(A)\leq N_{\text{min}}(K)$ we have that
$\mathrm{d_{p,min}}(L) \leq \sqrt{N(L)N_{\text{min}}(K)}$ and $\mathrm{Nd_{p, min}}(\psi(L))\leq  \frac{\sqrt{N_{\text{min}}(K)}2^{k/2}}{|d_K|^{1/4}}$. 

Assume that $S$ is an ideal such that $N(S)=N_{\text{min}}(K)$ and choose $I$ as an element from the class
$[S]^{-1}$. 
For any non-zero element $x \in I$, we then have that $x\OO_K=IC$, for some ideal $C$ that belongs to the class
$[S]$. Therefore we have that $n(\psi(x))\geq \sqrt{N(I)N(C)}$.
\end{proof}
This result translates the question of finding the product distance of an ideal into a well-known problem in algebraic number theory. It 
also 
suggests which ideal class we should use in order to maximize the product distance.

Denote with $\mathcal{K}_{2k}$ the set of totally complex number fields  of degree $2k$. 
Then the 
optimal normalized product distance over all complex fields of degree $2k$ and all ideals $I$ is
\begin{equation}\label{maximum}
\underset{K\in \mathcal{K}_{2k }}{\mathrm{max}} \frac{2^{k/2}\sqrt{N_{\text{min}}(K)}}{|d_K|^{\frac{1}{4}}}.
\end{equation}

As far as we know it is an open  question whether the maximum 
in \eqref{maximum} 
is  always achieved when $K$ is principal ideal domain.
Some preliminary data can be found 
in \cite{FOV}.  We point out that 
Proposition \ref{idealform} 
makes this problem computationally much
more accessible.

\section{Reduced Hermite invariants as homogeneous forms }\label{geometry}
Let us now see   how different linear channels define different  sets $\mathcal{A}_k$ and how the corresponding reduced norms can be seen as different \emph{homogeneous forms}. For simplicity we will study the case when we transmit four  information symbols $(x_1, x_2, x_3, x_4)$.

\smallskip

In the AWGN channel the receiver sees 
$$
(x_1, x_2, x_3, x_4)+ (w_1, w_2, w_3, w_4),
$$
where $w_i$ are Gaussian random variables. Here the set 
$\mathcal{A}_4^{(1)}$
simply consists of a single element, the $4\times4$ identity matrix. Therefore we obviously have 
$$
||(x_1,x_2,x_3,x_4)||_{\mathcal{A}_4^{(1)}}^2=|x_1|^2+|x_2|^2+|x_3|^2+ |x_4|^2.
$$

Let us then consider a channel where the fading stays stable for $2$ time units and then changes. Then the 
received signal will be of the form
$$
(h_1x_1, h_1x_2, h_2x_3, h_2x_4)+(w_1, w_2, w_3, w_4).
$$
Assuming that  $h_i$ are non-zero 
with probability $1$ we can see that 
$$
\mathcal{A}_4^{(2)}=\{\mathrm{diag}[a_1, a_1, a_2,a_2] \mid |a_1\cdot a_2|=1, a_i\in \C\}.
$$
Following the proof of  \cite[Proposition 8]{LV2017} we get  the following result
\begin{equation}\label{twoblock}
||(x_1, x_2,x_3,x_4)||_{\mathcal{A}_4^{(2)}}^2= 2\sqrt{(|x_1|^2 +|x_2|^2) \cdot (|x_3|^2+|x_4|^2)}.
\end{equation}
Earlier we considered  the fast fading channel in which the channel can change during every time unit giving us the following received vector:
$$
(h_1x_1, h_2x_2, h_3x_3, h_4x_4)+(w_1, w_2, w_3, w_4).
$$
In this case we have
that
\begin{equation}\label{diagonalmatrix}
\mathcal{A}_4^{(3)} =\{\mathrm{diag}[a_1, a_2, a_3,a_4] \mid |a_1\cdot a_2\cdot a_3 \cdot a_4|=1, a_i\in \C\}.
\end{equation}
and that
\begin{equation}\label{fast}
||(x_1,x_2,x_3,x_4)||_{\mathcal{A}_4^{(3)}}^2 = 4|x_1x_2x_3 x_4|^{1/2}.
\end{equation}

\smallskip
In all the previous examples the channel could be represented as a diagonal action. 
On the other hand, for a $2 \times 2$ MIMO system, the channel matrix will have 
 block diagonal structure. In this case the received vector can be written as
$$
 (h_1x_1+h_2x_2, h_3x_1+h_4x_2, h_1x_3+h_2x_4, h_3x_3+h_4x_4)+(w_1, w_2, w_3, w_4).
$$
Here the set 
$\mathcal{A}_4^{(4)}$ 
consists of matrices 
$$
\left\{ \begin{pmatrix}
h_1&h_2&0&0\\
h_3&h_4&0&0\\
0&0&h_1&h_2\\
0&0&h_3&h_4\\
\end{pmatrix}
\mid \,\mathrm{det} \left|
\begin{pmatrix}
h_1&h_2\\
h_3&h_4
\end{pmatrix}
\right|=1 \right\}
$$

According to \cite[Proposition 8]{LV2017} we have that
\begin{equation}\label{mimo}
||(x_1,x_2,x_3,x_4)||_{\mathcal{A}_4^{(4)}}^2=2|(x_1x_2-x_3x_4)|.
\end{equation}

We immediately note that all the reduced norms  share common characteristics.
\begin{definition}
A continuous function $F$: $\C^k \to \R$
is called a \emph{homogeneous form} of degree $\sigma>0$ if it satisfies the relation
$$
|F(\alpha {x})|=|\alpha|^{\sigma} |F(x)|\quad (\forall \alpha \in \R, \forall x \in \C^k).
$$
\end{definition}

Given a full lattice $L\in \C^k$ and assuming that $\Vol(L)=1$, we can  define the \emph{homogeneous minimum} of the form $F$ as 
\[
\lambda(F, L)=\inf_{x \in L \setminus \{ 0\}} \abs{F(x)}.
\]

Setting 
$||\,\, ||_{A_4^{(i)}}^2=F_{\mathcal{A}_4^{(i)}}$, 
we can see that each of the squared reduced norms defined previously are homogeneous forms of degree $2$.

As we saw in Theorem \ref{main}, given a sequence of random matrices $H_k$ of size $k \times k$ and the corresponding sets $\mathcal{A}_k$ in (\ref{normalized}), we can use $\mathrm{rh}_{\mathcal{A}_k}$ as a design criterion for building capacity-approaching lattice codes. In many cases of interest, $\norm{\;}_{\mathcal{A}_k}^2=F_{\mathcal{A}_k}$ will be a homogeneous form and $\mathrm{rh}_{\mathcal{A}_k}(L)=\lambda(F,L)$. For instance this is the case if we extend the previous examples to general size $k$ and define
\begin{align*}
& \mathcal{A}_k^{(1)}={I_k},\\
& \mathcal{A}_{2k}^{(2)}=\{\mathrm{diag}[a_1, a_1,a_2,a_2,\ldots, a_k,a_k] \mid |a_1 a_2 \cdots a_k|=1, a_i\in \C\},\\
& \mathcal{A}_k^{(3)} =\{\mathrm{diag}[a_1, a_2, \ldots ,a_k] \mid |a_1 a_2\cdots a_k|=1, a_i\in \C\}.
\end{align*}

In the case where $\mathcal{A}_k=\{I_k\}$, 
 we have recovered the classical connection between sphere packing and AWGN capacity, but we also proved that there exist similar connections between different channel models and the corresponding homogeneous forms.

A natural question is now how close to capacity we can get with these methods by taking the best possible lattice sequences in terms of their homogeneous minimum. We will denote with $\mathcal{L}_{k}$ the set of all the lattices $L$ in $\C^k$ having $\Vol(L)=1$.
This leads us to the concept of  \emph{absolute homogeneous minimum} 
$$
\lambda(F)=\sup_{L\in \mathcal{L}_{k} }\lambda(F,L).
$$
Finding the value of absolute homogeneous minima
is one of the central problems in geometry of numbers. As we saw earlier it is a central problem also in the theory of linear fading channels.

In the case $\mathcal{A}_k=\{I_k\}$, $\lambda(F_{\mathcal{A}_k})$ is the \emph{Hermite constant} $\gamma_k$.  The value of the Hermite constant 
for different values of $k$ has been studied in mathematics for hundreds of years and there exists an extensive literature on the topic. In particular good  upper and lower bounds are available and it has been proven that we can get quite close to Gaussian capacity with this approach \cite[Chapter 3]{CS}. 

In the case of $F_{\mathcal{A}_k^{(3)}}$
the problem of finding homogeneous minima has been considered in the context of algebraic number fields and some upper bounds have been provided.  Similarly for 
$F_{\mathcal{A}_k^{(2)}}$
there exists considerable literature. These and related results can be found in \cite{GL}. However for the case of homogeneous forms arising from block diagonal structures
there seems to be very little previous research.

\begin{remark}
While the definition of the reduced Hermite invariant is very natural, we have  found very few previous works considering similar concepts. The first reference we have been able to locate is  \cite{Skriganov}. There the author considered matrices of type \eqref{diagonalmatrix} and proved  \eqref{reduced} in this special case. Our results  can therefore be seen as a   natural generalization of this work.  The other relevant reference is  \cite{bayer} where the authors  defined the Hermite invariant for generalized ideals in division algebras in the spirit of \emph{Arakelov theory}. Again their definition is  analogous to ours in certain special cases. 
\end{remark}

\begin{remark}
We note that the reduced norms in our examples are not only homogeneous forms, but  multivariate polynomials and 
 the sets $\mathcal{A}_k^{(i)}$ are  groups. As we obviously have that
$$
||A(x)||_{\mathcal{A}_k^{(i)}}^2=||x||_{\mathcal{A}_k^{(i)}}^2,
$$
for any $A\in \mathcal{A}_k^{(i)}$, we can see that $||\,\,||^2_{\mathcal{A}_k^{(i)}}$ is actually a classical polynomial invariant of the group
$\mathcal{A}_k^{(i)}$.   At the moment we don't know what conditions a matrix group $\mathcal{A}_k$ should satisfy so that  the corresponding reduced norm would be a homogenous form. Just as well we don't know when  some power  of the reduced norm is   a polynomial.
\end{remark}




{\small

\end{document}